\documentclass{article}
\usepackage{arxiv}

\usepackage[utf8]{inputenc} 
\usepackage[T1]{fontenc}    
\usepackage{hyperref}       
\usepackage{url}            
\usepackage{booktabs}       
\usepackage{amsfonts}       
\usepackage{nicefrac}       
\usepackage{microtype}      
\usepackage{lipsum}		    
\usepackage{graphicx}
\usepackage{natbib}
\usepackage{doi}

\usepackage{appendix}

\RequirePackage{amsthm,amsmath,amsfonts,amssymb}
\RequirePackage{natbib}

\RequirePackage{bbold}

\usepackage{comment}
\usepackage[inline]{enumitem}

\theoremstyle{plain}

\newtheorem{theorem}{Theorem}[section]

\newtheorem{algo}{Algorithm}[section]
\theoremstyle{definition}

\newtheorem{assumption}{Assumption}
\theoremstyle{remark}

\newcommand{\expit}{{\text{expit}}}

\title{Posterior Uncertainty for Targeted Parameters in Bayesian Bootstrap Procedures}

\renewcommand{\shorttitle}{Bootstrap Methods for Posterior Uncertainty}

\author{ \hspace{1mm}Magid Sabbagh\\
	Department of Mathematics and Statistics\\
	McGill University\\
	Montreal, QC, Canada \\
	\texttt{magid.sabbagh@mail.mcgill.ca} \\
	\And
	\href{https://orcid.org/0000-0001-9811-7140}{\includegraphics[scale=0.06]{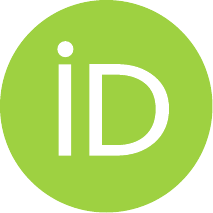}\hspace{1mm}David A. Stephens}\thanks{Corresponding author} \\
	Department of Mathematics and Statistics\\
	McGill University\\
	Montreal, QC, Canada \\
	\texttt{david.stephens@mcgill.ca} \\
}

\begin{document}

\maketitle

\begin{abstract}
We propose a general method to carry out a valid Bayesian analysis of a finite-dimensional `targeted' parameter in the presence of a finite-dimensional nuisance parameter. We apply our methods to causal inference based on estimating equations. While much of the literature in Bayesian causal inference has relied on the conventional 'likelihood times prior' framework, a recently proposed method, the 'Linked Bayesian Bootstrap', deviated from this classical setting to obtain valid Bayesian inference using the Dirichlet process and the Bayesian bootstrap. These methods rely on an adjustment based on the propensity score and explain how to handle the uncertainty concerning it when studying the posterior distribution of a treatment effect. We examine theoretically the asymptotic properties of the posterior distribution obtained and show that our proposed method, a generalized version of the 'Linked Bayesian Bootstrap', enjoys desirable frequentist properties. In addition, we show that the credible intervals have asymptotically the correct coverage properties. We discuss the applications of our method to mis-specified and singly-robust models in causal inference.
\end{abstract}


\section{Introduction}
In an inference problem, focus is typically placed on a quantity that has some meaningful interpretation in terms the observable quantities.  In the classical Bayesian representation, such quantities are typically the parameters defined by de Finetti's representation.  However, we may take a more general approach to `targeting' parameters by defining them as functionals of the data generation mechanism, arrived at, for example, as the solution to a loss minimization problem. This more general kind of definition of a targeted parameter does not always lend itself to conventional Bayesian inference as such parameters might not appear in the likelihood in a straightforward manner.  It is usually the case that such targeted parameters do not specify the data generating process fully.  We adopt the term `targeted parameter' rather than the more common `parameter of interest' as we feel that it better reflects the inference objective.

Targeted parameters also arise in situations where we are conscious of model mis-specification. For example, in clinical studies we might target a parameter which quantifies the effect of a treatment on an outcome. In this causal inference setting, the key role played by the propensity score in observational studies is highlighted in \citet{Rosenbaum/Rubin:1983}. Adjustment based on the propensity score can be carried out using weighting, stratification, matching or regression. While the use of a fitted propensity score in the frequentist literature is common in the study of the causal effect, it is not universally accepted that fully Bayesian inference using this plug-in is possible. \citet{Stephens/etal:2022} address the issue of incorporating a parametrically modeled propensity score in a Bayesian causal analysis based on the Dirichlet process in a possible model mis-specification. While there was no consensus that two-step methods based on the fitted values of the propensity score can be viewed as fully Bayesian methods, the computational method they introduce provides a valid Bayesian inference on a causal targeted parameter. The aim of this paper is to give a theoretical justification of this computational method, and show that indeed, the resulting posterior exhibits coverage at the nominal level asymptotically.  We adopt a general framework that uses ideas from empirical processes, the Bayesian bootstrap, and estimating equations. 

In sections \ref{Background}, we outline a bootstrap method derived from applications of the Dirichlet process and the Bayesian bootstrap.  In section \ref{ExtensiontoEstimatingEquations} we show how targeting parameters via estimating equations, and use these notions to formulate the problem of obtaining fully valid Bayesian inference of a targeted parameter in presence of an estimated nuisance parameter..  In section \ref{Asymptotics} we derive the main theoretical results of this paper.  We study the effect of handling the inferential uncertainty concerning the nuisance parameter on the posterior of the targeted parameter, and theoretically demonstrate the validity of a recently proposed method, the Linked Bayesian Bootstrap. We move forward by giving an argument validating the use of two-steps procedures based on the propensity score in causal structural models as Bayesian methods. In section \ref{Simulations}, we empirically validate the results we have previously obtained, performing analysis on treatment effects using G-estimation and inverse weighting estimators.


\section{Bootstrap Methods and Loss Minimization}{\label{Background}}

In the following we adopt the following notation: let $O_1,\ldots,O_n$ be i.i.d random variables defined on a probability space $(\mathcal{O},B)$ with distribution $F_0=P_O$. Denote and define the empirical process $\mathbb{P}_n$ and weighted empirical process $\mathbb{P}_n^w$ as
\[
\mathbb{P}_n=\frac{1}{n}\sum\limits_{i=1}^n \delta_{O_i} \qquad \text{and} \qquad \mathbb{P}_n^w=\sum\limits_{i=1}^n w_{in}\delta_{O_i}
\]
respectively, where $\{w_{in}, i=1,\ldots,n\}$ is a series of (possibly stochastic) weights summing to one.   In this paper, we focus exclusively on Dirichlet weights.


\subsection{The Loss-Likelihood Bootstrap}{\label{Loss-LikelihoodBootstrap}}

One appealing approach to targeting parameters uses loss minimization.  Suppose the true value $\theta_0 \in \Theta \subset \mathbb{R}^p$ of a parameter $\theta$ minimizes the following expected loss, for some loss function $l :  \mathcal{O} \times \Theta \to [0,\infty)$ :
$$
\theta_0  \equiv  \theta(F_O)= \text{argmin}_{\theta} \int_{\mathcal{O}} l(o;\theta)dF_0(o) =\text{argmin}_{\theta}\text{ }\mathbb{E}_{P_0} l(O;\theta).
$$
If measure $P_O$ with corresponding distribution function $F_O$ is unknown but is known to belong to a set $\mathcal{F}$, any prior $P_\mathcal{F}$ on $\mathcal{F}$ induces a prior on $P_\Theta$ on  $\Theta$, by first considering, for $F \in \mathcal{F},$
$$
\theta(F)= \text{argmin}_{t} \int_{\mathcal{X}} l(o;t)dF(o)
$$
and then defining $P_\Theta (\theta \in A) = P_\mathcal{F} \left\{F: \theta(F) \in A \right\}$.  Under mild conditions on $l$ and prior $P_{\mathcal{F}}$, this prior places sufficient mass in neighbourhoods of $\theta_0$.

\citet{Lyddon/etal:2019} use the loss minimization approach to targeting in order to obtain a sample from the posterior distribution of $\theta(P_O)$ using a Bayesian non-parametric computational strategy. Let $\alpha$ be a finite Borel measure on $(\mathcal{O},B)$. If a Dirichlet process prior DP$(\alpha)$, see \citet{Ferguson:1973}, is placed on $F$, then the posterior distribution of $F$ given $O_1,\ldots,O_n$ will follow a Dirichlet process  DP$(\alpha+n\mathbb{P}_n)$.  The weak limit DP($n\mathbb{P}_n)$, which can represented by the empirical process $\mathbb{P}_n^w$ with $(w_{1n},\ldots,w_{nn}) \sim \text{Dir}(1,\ldots,1)$, of the posterior as $\alpha(\mathcal{O}) \to 0$ is known as the Bayesian bootstrap distribution; see \citet{Rubin:1981}. It is also worth mentioning that the posterior distribution DP$(\alpha + n \mathbb{P}_n)$ can be written as $V_n Q +(1-V_n) \mathbb{P}_n^w$, where $V_n \sim \text{Beta}(\vert \alpha \vert,n)$ and $Q \sim \text{DP}(\alpha)$  with  $\vert \alpha \vert =\alpha(\mathcal{O})$ and $V_n,Q$ and $w$ are independent, see \citet{Ghosal/VanderVaart:2017}. When $n$ is large, it emerges that the $\mathbb{P}_n^w$ is the dominating term and this fact lies at the heart of the non-parametric Bernstein-von Mises theorem. Combining these facts about Dirichlet processes leads to approximating the posterior of $\theta(F)$ by 
$$
\text{argmin}_\theta \int_{\mathcal{X}} l(O;\theta) \ d\mathbb{P}_n^w = \text{argmin}_\theta \sum\limits_{i=1}^n w_{in}l(O_i;\theta).
$$  

\noindent Posterior computation can be achieved using the following algorithm:
\begin{algo}{Loss-Likelihood Bootstrap \citep{Lyddon/etal:2019}\label{algoLLB}}
\begin{enumerate}
    \item For $j \in \{1,2,\ldots,B\}$
\begin{enumerate}
    \item Draw random weights $(w^{(j)}_{1n},\ldots,w^{(j)}_{nn}) \sim \text{Dir}(1,\ldots,1)$.
    \item Compute $\displaystyle{ \theta^{(j)}={arg\,min}_\theta \sum\limits_{i=1}^n w_{in}^{(j)}  l(O_i;\theta)}$.
\end{enumerate}
    \item Output $\theta^{(1)}.\ldots,\theta^{(B)}$.
\end{enumerate}

\end{algo}
The following theorem of \citet{Lyddon/etal:2019} generalizes the result of \citet{Newton/Raftery:1994} and concerns the conditional (on the data) asymptotic distribution of a bootstrap sample for the minimizer of $\mathbb{E}_{P_O} l(O;\theta)$, $\theta_0$, where $O$ has distribution $P_O$. We define the necessary assumptions on the relevant probability spaces in Assumption \ref{ass0} (see the supplementary material). In the remaining parts of this paper, weights $(w_{1n},\ldots,w_{nn})$ follow a $\text{Dirichlet}(1,\ldots,1)$ distribution, denoted by $P_W$ and are generated to be independent of the data. The joint probability measure is hence $P_{OW}=P_{O}^\infty \times P_W$. We write $P_O$ in lieu of $P_O^{\infty}$ for simplicity when necessary.  
\begin{theorem}{\label{thmLLB}}
   Let $\hat{\theta}_n$ and $\hat{\theta}_{n,BB}$ be the minimizers of
   \[
   \frac{1}{n} \sum\limits_{i=1}^n l(O_i;\theta)  \qquad \text{and} \qquad \sum\limits_{i=1}^n w_{in} l(O_i;\theta),
   \]
   respectively, where $(w_{1n},\ldots,w_{nn})$ is a Dirichlet$(1,\ldots,1)$ random vector. Then, under regularity conditions, the quantities
   \[
\sqrt{n} ( \hat{\theta}_n-\theta_0 ) \qquad \text{and} \qquad 
   \sqrt{n}  (\hat{\theta}_{n,BB}-\hat{\theta}_n )\mid O_1,\ldots O_n
   \]
   converge in distribution as $n \to \infty$ to a random variable with distribution
   \[
   \mathcal{N}\left(0, J(\theta_0)^{-1} I(\theta_0) J\left(\theta_0\right)^{-1}\right)
   \]
   almost surely $P_O^\infty$, where
   \[
   J(\theta_0)= \mathbb{E}_{P_O}\left\{ \frac{\partial^2l(O;\theta_0)}{\partial \theta \partial \theta^\top}\right\} \quad and  \quad I(\theta_0)=  \mathbb{E}_{P_{O}} \left\{  \frac{\partial l(O;\theta_0)}{\partial \theta}\frac{\partial l(O;\theta_0)}{\partial \theta^\top}\right\}.
   \]
   where
   \[
    \frac{\partial l(o;\theta_0)}{\partial \theta} =  \left.\frac{\partial l(o;\theta)}{\partial \theta} \right |_{\theta = \theta_0}
   \]
   and so on.
\end{theorem}

The conditional asymptotic distribution from Theorem \ref{thmLLB} is, in this case, the same as the limiting distribution of $\sqrt{n} ( \hat{\theta}_n-\theta_0)$. This is known as Bayesian bootstrap consistency. It is critical, however, to highlight the contrast between the two limiting statements. The randomness in the first comes through the random variables $O_1\ldots,O_n$, whereas in the second these variables are assumed fixed, and the only source of randomness is the weights. This Bayesian/frequentist duality is of major importance in interpreting Bayesian credible regions and in assessing the asymptotic frequentist properties of the posterior. The Bayesian credible interval around the posterior mean $\hat{\theta}_n$ is equivalent to the frequentist confidence interval based on $\hat{\theta}_n$. In essence, the conditional behavior of the posterior around its mean will mimic the frequentist behavior of the posterior mean around the true value of the parameter. Furthermore, the Bayesian confidence set will have (asymptotically) the correct coverage probability. In Bayesian analysis, the parametric Bernstein-von Mises theorem establishes such a bridge between the maximum \textit{a posteriori} mean and the maximum likelihood estimator, see \citet{Ghosh/Ramamoorthi:2003}. Such duality is studied and known to hold for semi-parametric M-estimators; see \citet{Cheng/Huang:2010} and \citet{Kosorok/Ma:2005}.
\subsection{Loss-Likelihood Bootstrap in Multi-Parameter Settings}
When $\theta$ is a multicomponent vector, we may investigate the posterior distribution of a sub-component of $\theta$ by marginalizing in the usual way.  Specifically, the marginal posterior (asymptotic) variance can be computed by taking the relevant block of the joint covariance matrix $J(\theta_0)^{-1} I(\theta_0) J\left(\theta_0\right)^{-1}$.  Similarly, the variance of the (asymptotic) conditional distribution for one sub-component given another can be computed using the usual properties of the multivariate Normal distribution. In this case, by standard arguments, it is evident that the variance of the marginal posterior is always at least as large as the conditional posterior variance, in that the difference between the former and the latter is non-negative definite.

The loss-likelihood bootstrap is explicitly formulated as a joint (weighted) minimization over all elements of $\theta$.  That is, $\theta_0$ is an explicit target of the inference, and the loss function encapsulates the relationship between the observed data and all elements of $\theta$ simultaneously.   However, as the joint distribution is a valid probability distribution, marginals resulting from it are also valid distributions.  Thus an interpretation of the computed marginal posterior in isolation is not straightforward; this is in part a consequence of the absence of a generative probability model for the observable quantities.  

If $\theta^\top = (\theta_1^\top,\theta_2^\top)$, the formulation
\[
\theta_{10}(F)= \text{argmin}_{\theta_1} \int_{\mathcal{X}} l(o;\theta_1,\theta_{20})dF(o)
\]
is natural, but requires a further step of defining $\theta_{20}$.  Fundamentally, the two parameters are intrinsically linked through the loss function, and identification of the targeted parameter through some marginalized approach is not straightforward.  It may also be the case that $\theta_{10}$ is not identified by using a uniquely defined loss function, and that two different loss functions identifying the same parameter may not yield the same marginal inferences.  Specifically, consistent estimation may still be possible, but different (asymptotic) variances may be obtained.

If the loss function is separable, say
\[
l(o;\theta_1,\theta_2) = l_1(o;\theta_1) + l_2(o;\theta_2)
\]
then marginal interpretations for $\theta_1$ and $\theta_2$ may be possible.  Separability arises in the conventional parametric setting under exchangeability and de Finetti's representation for a hierarchical model, say where
\[
f_1(y|x;\theta_1) f_2(x;\theta_2)
\]
and where $\theta_1$ and $\theta_2$ have interpretations in the conditional and marginal models respectively.  If there is no separability, then it is possible that a profiling approach may be adopted; if $\theta_1$ is the target, then for each $\theta_1$ we consider minimizing for $\theta_2$, and examining the profiled loss function as $\theta_1$ varies.  However, this leaves the question of where and how the Dirichlet weights should be utilized.  We attempt to resolve this point in the following sections.

\section{Bootstrap Approaches for Estimating Equations}{\label{ExtensiontoEstimatingEquations}}
The loss-likelihood bootstrap focuses on minimization of a loss function, but there are obvious links an approach based on root-finding, as in the frequentist approach based on estimating equations.  In this section, we adapt the arguments discussed above to parameters defined as solutions to estimating equations. We then discuss extensions to estimating equations involving a finite-dimensional nuisance parameter, and motivate the problem of obtaining a sample from the posterior distribution of the targeted parameter that possesses good frequentist properties.

\subsection{An Illustrative Model}{\label{AnIllustrativeModel}}
Let $O=(Y,X,Z,U_1,U_2)$ and let $O_1,\ldots,O_n$ be i.i.d random variables. Suppose that the data generating mechanism is
\begin{eqnarray}\label{plm}
  Y & = & \theta_0 Z +g_0(X) +U_1,\\
  Z & = & e_0(X)+U_2,
\end{eqnarray}
where $\theta_0 \in \mathbb{R}$, $g_0$ and $e_0$ are real-valued functions of $X \in \mathbb{R}^q$, $U_1 \sim N(0,\sigma^2)$, and $\mathbb{E}[U_2|X]=0$.  The model in \eqref{plm} is called the \textit{partially linear model}. Flexible estimation of $g_0$ with minimal assumption on its form is in general not an obvious task.

Suppose that $Z$ is a Bernoulli random variable with $P(Z=1|X)=e_0(X)$; in this case $U_2 = \mathbb{1}\{U < e_0(X) \} - e_0(X)$ where $U \sim Uniform(0,1)$ is independent of $U_1$ and $X$, so that $U_2$ is uncorrelated with $e_0(X)$ and has variance $e_0(X) (1-e_0(X))$.  In the causal inference setting, the function $e_0$ is known as the propensity score and $\theta_0$, known as the average treatment effect (ATE), represents the targeted parameter. Suppose further that $e_0(x)=e(x;h_0)=\expit(h_0^\top x)$, where $\expit(t)=e^t/(1+e^t)$ and $h_0$ is a nuisance parameter in $\mathbb{R}^q$.  Defining
\begin{equation}\label{RMN}
m(O;\theta_0,h_0)=(Y-\theta_0 Z)\left\{Z-e(X;h_0)\right\},
\end{equation}
\citet{Robins/etal:1992} use the fact that $\mathbb{E} m(O;\theta_0,h_0) =0$ 
to obtain an estimator of $\theta_0$. Since $h_0$ is not known in practice, it is estimated by performing a logistic regression of $Z$ on $X$, which results in a estimator $\hat{h}$ of $h_0$. An estimator $\hat{\theta}_n$ of $\theta_0$ depending on $\hat{h}$ is then obtained as the solution to
\[
\sum_{i=1}^n (Y_i-\theta Z_i)\left\{Z_i-e(X_i;\hat{h})\right\}=0,
\]
so that
\begin{equation} \label{theta-G-est}
\hat{\theta}_n= \frac{\sum\limits_{i=1}^n Y_i\left\{Z_i-e(X_i;\hat{h})\right\}}{\sum\limits_{i=1}^n Z_i\left\{Z_i-e(X_i;\hat{h})\right\}}.
\end{equation}

There are other possible mis-specified models that can be of interest. For example, we might consider the moment restriction $\mathbb{E} k(O;\theta_0,\phi_0,h_0)=0$, where
\begin{equation}\label{psree}
k(O;\theta,\phi,h)= \left\{Y- \theta Z - \phi e(X;h)\right\}\left\{Z-e(X;h)\right\},
\end{equation}
and $\phi_0=\mathbb{E}\left\{g_0(X)e_0(X) \right\}/\mathbb{E}\left\{e_0(X)^2\right\}$. Solving for $\psi$ and $\phi$ in terms of $h$ is equivalent to regressing the outcome $Y$ on the treatment $Z$ and the propensity score. This approach attempts to account for $g_0(X)$ by projecting it onto the span of $e_0(X)$ whereas the function $m$ used previously does not account for $g_0(X)$. We refer to \citet{Moodie/Stephens:2022} for details on the mis-specified model described by $k$, which will be called the propensity score regression model.   Note that, under certain conditions, the functions in \eqref{RMN} and \eqref{psree} yield consistent and asymptotically normal (frequentist) estimators.

There are other estimators of treatment effects that rely exclusively on knowledge of the propensity score. Such an estimator is the IPW estimator of $\theta_0$. This estimator stems from the moment restriction
\begin{equation}\label{ipwee}
\mathbb{E}\left[\theta_0 -\left\{\frac{ZY}{e(X;h_0)}-\frac{(1-Z)Y}{1-e(X;h_0)}\right\}\right] =0.
\end{equation}
Another such estimator is the estimator of the average treatment effect on the treated (ATT). The  ATT, $\rho_0$, defined by $\rho_0=\mathbb{E}\left\{Y(1)-Y(0) \vert Z=1 \right\}$, satisfies the moment restriction $\mathbb{E} r(O;\rho_0,h_0)=0$, where
\[
r(O;\rho,h)= ZY-\frac{Y(1-Z)e(X;h)}{1-e(X;h)}-\rho Z.
\]
We refer to \citet{Moodie/etal:2018} for an overview of the ATT. The resulting estimator is not affected by the form of the outcome model but is contingent only upon the correct specification of the propensity model. Such a moment restriction cannot be seen as a full specification of the likelihood, but rather a constraint on the distribution of the observables, determined by the general assumptions in the causal inference framework.

\subsection{Targeting Parameters using Estimating Equations}
{\label{NoNuisanceParameter}}

Let $m: \mathcal{O} \times \Theta \to \mathbb{R}^p$.  For $F \in \mathcal{F}$, let $\theta(F)$ denote the solution to
\[
\int_\mathcal{O}m(o;\theta)dF(o)=0.
\]
and let $\theta_0 \equiv \theta(F_O)$. Using the same arguments discussed in section \ref{ExtensiontoEstimatingEquations}, we can obtain a sample from the posterior distribution of $\theta$ by implementing the following algorithm:
\begin{algo}{Bayesian Bootstrap for Estimating Equations} {\label{algononuisance}}
 \begin{enumerate}
        \item For each $j \in \left\{1,\ldots,B\right\}$,
        \begin{enumerate}
        \item  Draw random weights $(w^{(j)}_{1n},\ldots,w^{(j)}_{nn}) \sim \text{Dir}(1,\ldots,1)$.
        \item  Find $\theta^{(j)}$ as the solution to
               $
               \sum\limits_{i=1}^n w^{(j)}_{in} m(O_i;\theta)=0 
               $
        \end{enumerate}
        \item Output $\theta^{(1)},\ldots, \theta^{(B)}$. 
    \end{enumerate}
\end{algo}
The following standard result, see for example \citet{Kosorok:2008}, establishes the necessary Bayesian/frequentist duality for Algorithm \ref{algononuisance}. This theorem can also be proved using the same techniques and comparable regularity conditions to those that will be used in Theorem \ref{ThmEstimatedNuisance} and Theorem \ref{ThmLBBConsistency}, so we defer discussion until later.
\begin{theorem}\label{thmllbroot}
    Let $\hat{\theta}_n$ and $\hat{\theta}_{n,BB}$ satisfy 
    \[
    \sum\limits_{i=1}^n m(O_i;\hat{\theta}_n)=0 \qquad \text{and} \qquad  \sum\limits_{i=1}^n w_{in} m(O_i;\hat{\theta}_{n,BB})=0
    \]
    respectively.  Then, under regularity conditions,
\[
   \sqrt{n} ( \hat{\theta}_n-\theta_0 ) 
 \qquad \text{and} \qquad   \sqrt{n} (\hat{\theta}_{n,BB}-\hat{\theta}_n)|O_1,\ldots O_n
   \]
   converge in distribution as $n \to \infty$ to a random variable with a $\mathcal{N}(0, L)$ distribution, where 
   \[
   L= \left[\mathbb{E}_{P_O}\left\{ \frac{\partial{m(O;\theta_0)}}{\partial \theta^\top}\right\}\right]^{-1} \mathbb{E}_{P_O}\left\{ m(O;\theta_0)m(O;\theta_0)^\top \right\}\left[\mathbb{E}_{P_O}\left\{ \frac{\partial{m(O;\theta_0)}}{\partial \theta^\top}\right\}\right]^{-\top}
   \]
\end{theorem}
\noindent The second result yields a fully Bayesian (asymptotic) posterior distribution for the targeted parameter.  If $l$ from Theorem \ref{thmLLB} is differentiable with respect to $\theta$, then the minimization problem is equivalent to that addressed in Theorem \ref{thmllbroot}.  However, adopting the estimating equations viewpoint allows for some more generality. Differentiability of $m(O;\theta)$ can be replaced by the differentiability of $\mathbb{E}_{P_0}m(O;\theta)$ provided that $m$ belongs to a Donsker class of functions, allowing the use of the properties of the exchangeable bootstrap for Donsker classes; see \citet{VanderVaart/Wellner:1996}. 

\subsection{Targeting in the Presence of a Nuisance Parameter}
We now assume the presence of a finite-dimensional nuisance parameter $h_0$, and we assume that $\theta_0$ is the solution to 
\[
\int_{\mathcal{O}}m(o;\theta,h_0)dF_O(o)=0.
\]
In the multi-parameter setting, this is a natural way to facilitate targeting of $\theta_0$. In an analogous fashion to the methods previously discussed, we can obtain a posterior distribution by sampling Dirichlet weights.  In practice however, $h_0$ is unknown and needs to be estimated. If $\tilde{h}$ is an estimator of $h_0$, the exact method for handling the uncertainty of $\tilde{h}$ when obtaining a sample from the posterior of the targeted parameter is a notion of utmost importance. We seek a posterior distribution that satisfies a frequentist/Bayesian duality theorem for the reasons mentioned in subsection \ref{Loss-LikelihoodBootstrap}. 

Many possible approaches may be adopted when tackling the problem of incorporating $\tilde{h}$ in the posterior inference for $\theta$, and several questions arise. First, should $\tilde{h}$ be obtained in a separate first step and remain fixed throughout the posterior calculation? In other words, the posterior distribution of $\theta$ obtained using the Bayesian bootstrap is conditional on the knowledge of $\tilde{h}$, which remains unchanged during the Bayesian procedure.  Here, no attempt is made to propagate the uncertainty attached to $\tilde{h}$. What are the properties of the resulting posterior for the targeted parameter, and how does this posterior depend on the way $\tilde{h}$ is obtained?  Secondly, if we were to account for the uncertainty attached to $\tilde{h}$, how should it be handled in order to obtain a posterior with good frequentist properties?

We address these questions section \ref{Asymptotics}.   Our focus is on Bayesian/frequentist duality, and is therefore asymptotic in nature.  However, we also study finite sample implications.  We note that Bayesian approaches based on estimating equations with different specifications may yield posterior distributions with different variances.  For example, in the partially linear model, attempts to represent $g_0$ using a flexible parametric family will usually lead to a smaller posterior variance than that computed for estimator \eqref{theta-G-est}.  Such results are well-studied in the frequentist literature, and our results allow us to draw similar conclusions in the Bayesian setting.

\section{Asymptotic Results}{\label{Asymptotics}}

Returning to the illustrative model of section \ref{AnIllustrativeModel}, the estimator in equation \eqref{theta-G-est} has tractable properties. The frequentist theory of estimating equations enables us to conclude that $\sqrt{n} (\hat{\theta}_n-\theta_0 )$ is asymptotically normal with an asymptotic variance $V$, see \citet{Tsiatis:2006}. \citet{Robins/etal:1992} remark that $V$ is smaller than asymptotic variance $\Sigma$ of the estimator of $\theta_0$ had $h_0$ been assumed to be known and not estimated. This paradoxical phenomenon, due to the fact that adding variability to $h_0$ reduces variability in $\theta_0$, is generalized by \citet{Henmi/Eguchi:2004} who show that an estimated nuisance parameter may lead to an improved asymptotic variance of the estimator of the targeted parameter.

In the next sections, we will deviate from this frequentist framework in order to conduct Bayesian inference concerning the targeted parameter $\theta_0$ in presence of a nuisance parameter $h_0$. 

\subsection{The Linked Bayesian Bootstrap}{\label{ApproachestoBCI}}
In the frequentist framework, adjustments based on the propensity score are typically carried out using regression, matching, inverse weighting or stratification. Two-step procedures, in which an estimated propensity score is used to infer about treatment effects, are widely studied. In a Bayesian framework, a two-step procedure is contrary to the conventional Bayesian framework of 'likelihood times prior'. A detailed survey of existing parametric approaches and challenges in Bayesian causal inference can be found in \citet{Li/etal:2022}. Methods that assume a joint parametric model for the treatment effect and propensity scores, see \citet{Mccandless/etal:2009}, produce feedback leading to biased estimates. Other approaches to remedy this feedback by the cutting feedback method; see \citet{McCandless/etal:2010} and two-step approaches in which Bayesian inference is performed conditional on the knowledge of an estimate of the propensity score; see \citet{Hoshino:2008} and \citet{Kaplan/Chen:2012}. While there has been no consensus that the two-step and cutting feedback approaches can be seen as fully Bayesian, \citet{Stephens/etal:2022} deviate from the Bayesian parametric framework and use Bayesian non-parametrics methods in order to infer about a finite-dimensional parameter. They argue that two-step methods are valid Bayesian procedures in causal structural models on the condition that a Bayesian estimate of the propensity score is carefully obtained.

In the context of the propensity score regression model given by the estimating function \eqref{psree} in subsection \ref{AnIllustrativeModel}, it is possible to obtain a posterior sample  of  $(\theta,\phi)$ by the following procedure, termed the Linked Bayesian Bootstrap: 
\begin{algo}{Linked Bayesian Bootstrap}{\label{LBBOpt}}
     \begin{enumerate}
        \item For each $j \in \left\{1,\ldots,B\right\}$,
        \begin{enumerate}
        \item  Draw random weights $(w^{(j)}_{1n},\ldots,w^{(j)}_{nn}) \sim \text{Dir}(1,\ldots,1)$.
        \item  Let $h^{(j)}=argmin_{h} \sum\limits_{i=1}^n w_{in} \log f_{2}(z_i|x_i;h) $
        \item  Let $(\theta^{(j)},\phi^{(j)})$ be defined as
               \[
               (\theta^{(j)},\phi^{(j)})= argmin_{(\theta,\phi)}\sum\limits_{i=1}^n w^{(j)}_{in} \log f_{1}(y|x,z,\theta,\phi, \hat{h}^{(j)})
               \]
        \end{enumerate}
        \item Output $(\theta^{(1)},\phi^{(1)}),\ldots, (\theta^{(B)},\phi^{(B)})$. 
    \end{enumerate}
\end{algo}

\noindent where $f_2(z|x)$ is the correctly specified conditional density of $Z|X$ and $f_1(y|x,z)$ is the consciously mis-specified density of $Y|X,Z$. The term 'Linked' is to emphasize that a single set of Dirichlet weights is retained at both optimization steps and provides the correct way to propagate the inferential uncertainty of $h$ in order to obtain the posterior of the targeted parameter  $(\theta,\phi)$.

In the remainder of this section, we will examine the Linked Bayesian Bootstrap through other lenses. After theoretically assessing its favorable frequency-based properties, we will interpret the Linked Bayesian Bootstrap in terms of estimating equations rather than loss-minimization, allowing us to view it as a one-step procedure, that simultaneously allows the sampling from the joint posterior distribution of $\theta$ and $h$. 

\subsection{Conventional Plug-In Method}{\label{Disadvantagesofnotadjustingfortheuncertainty}}
In this section, we examine the effect of not propagating the uncertainty attached to the nuisance parameter estimator in a two-step procedure. For example, suppose that $\hat{h}$ is an estimator of $h_0$ that is asymptotically linear, that is 
\[
\sqrt{n} (\hat{h}-h_0  ) = \frac{1}{\sqrt{n}} \sum_{i=1}^n \psi(O_i;h_0) +o_{P_{O}}(1),
\]
where $\mathbb{E}_{P_O} \psi(O;h_0)=0$ and $\text{Var}_{P_0}\psi(O;h_0)< \infty$.
The core idea behind this conventional plug-in approach lies in eliminating feedback with an accurate estimate of the nuisance parameter, see \citet{Zigler/etal:2013}. \citet{Saarela/etal:2015} point out that frequentist two-step procedures, that often rely on an adjustment based on the propensity score, do not have Bayesian counterparts and cannot be viewed as fully Bayesian, in the context of the 'likelihood times prior' conventional Bayes approach. \citet{Graham/etal:2016} propose a variance correction in order to accommodate the inferential uncertainty of the propensity score in two-step procedures for doubly robust models. 

We show that not accounting for the inferential uncertainty attached to $\hat{h}$ in estimating $\theta$ leads to undesirable behavior of the posterior estimated through the Bayesian bootstrap, regardless of the estimator $\hat{h}$. While classical frequentist parametric results incite the use of an estimated nuisance parameter in order to reduce the variability of the targeted parameter as explained in subsection \ref{AnIllustrativeModel}, our Bayesian calculation reveals that the asymptotic posterior variance is unaffected by a fixed estimate. Moreover, we show the surprising result that the asymptotic posterior variance does not depend on $\hat{h}$ and that $\hat{h}$ influences only the posterior mean. More specifically, the asymptotic posterior variance in this case is equal to the asymptotic posterior variance had $h_0$ been assumed to be known throughout. Such a result may seem advantageous, nut we expand upon its drawbacks in the discussion following the statement of Theorem \ref{ThmEstimatedNuisance}.

There are several possible estimators $\hat{h}$ that can be considered: for example, if $h_0$ is known, one can take $\hat{h}=h_0$. If $h_0$ is unknown we can obtain $\hat{h}$ by a regular frequentist logistic regression of the treatment on the covariates. From a Bayesian perspective, it may be more adequate perform a Bayesian analysis on the nuisance parameter. Estimators of $h_0$ of Bayesian nature can still be considered by first obtaining a sample from the posterior of $h$, through the loss-likelihood bootstrap for example, and then by letting $\hat{h}$ be the average of the sample obtained. One can show that these three suggested estimators have different influence functions with different variances, and yet the asymptotic posterior variance of the targeted parameter obtained without accounting for the inferential uncertainty of the point estimate of $h_0$ remains unchanged. 

In the illustrative model of Section \ref{AnIllustrativeModel}, the following computational approach may be adopted:
\begin{algo}{Bayesian Bootstrap for the Conventional Plug-in Method}{\label{algonuisancefixed}}
     \begin{enumerate}
        \item Obtain an estimator $\hat{h}$ of $h_0$ that is asymptotically linear.
        \item For each $j \in \left\{1,\ldots,B\right\}$,
        \begin{enumerate}
        \item  Draw random weights $(w^{(j)}_{1n},\ldots,w^{(j)}_{nn}) \sim \text{Dir}(1,\ldots,1)$.
        \item  Find $\theta^{(j)}$ as the solution to
               \[
               \sum\limits_{i=1}^n w^{(j)}_{in}\left\{Y_i-\theta Z_i\right\}\left\{Z_i-e(X_i;\hat{h})\right\}=0 
               \]
        \end{enumerate}
        \item Output $\theta^{(1)},\ldots, \theta^{(B)}$. 
    \end{enumerate}
\end{algo}

To compute the posterior variance, we resort to rudimentary tools from empirical processes. We will introduce $\hat{\theta}_n$ as the solution to the following estimating equation
\[
\sum\limits_{i=1}^n m(O_i,\theta,\hat{h})=0.
\]
In practice, $\hat{\theta}_n$ is a frequentist estimator of $\theta_0$ obtained using a plug-in estimator $\hat{h}$ of $h_0$. Even though it may be counterintuitive or unusual to use a Bayesian estimator of $h_0$ in a first step in order to obtain a frequentist estimator of $\theta_0$, we underline that $\hat{\theta}_n$ is introduced solely as a theoretical tool that facilitates identifying the asymptotic posterior distribution. We elaborate on this matter after stating Theorem \ref{ThmEstimatedNuisance}. The following theorem formalizes the results discussed above:
\begin{theorem}{\label{ThmEstimatedNuisance}}
    Suppose that $\hat{h}$ is an estimator of $h_0$ satisfying
        \[
        \sqrt{n}(\hat{h}-h_0)= \displaystyle{\frac{1}{\sqrt{n}}} \sum\limits_{i=1}^n \psi(O_i;h_0)+o_{P_{O}}(1),
        \]
    where $\mathbb{E}_{P_O} \psi(O;h_0)=0$ and $\text{Var}_{P_O} \psi(O;h_0)<\infty$. Let $\hat{\theta}_n$ and $\hat{\theta}_{n,BB}$ satisfy \[
    \sum\limits_{i=1}^n m(O_i,\hat{\theta}_n,\hat{h})=0 \qquad \text{and} \qquad \sum\limits_{i=1}^nw_{in} m(O_i,\hat{\theta}_{n,BB},\hat{h})=0,
    \]
    respectively.  Let
    \[
    M_{\theta_0} = \mathbb{E}_{P_O}\left\{  \frac{\partial{m(O;\theta_0, h_0)}}{\partial \theta^\top } \right\} \equiv \mathbb{E}_{P_O}\left\{ \left. \frac{\partial{m(O;\theta, h_0)}}{\partial \theta^\top } \right |_{\theta = \theta_0} \right\}.
    \]
    Then, under the regularity conditions stated in Assumption \ref{ass2}, the following hold:
 \begin{enumerate}
\item[(a)] 
\[
\sqrt{n}\left(\hat{\theta}_n-\theta_0\right)=\frac{-1}{\sqrt{n}}M_{\theta_0}^{-1}\sum\limits_{i=1}^n  \left\{m(O_i;\theta_0,h_0) + M_{h_0}\psi(O_i;h_0)\right\}+o_{P_O}(1),
\]
\item[(b)]
\[
\sqrt{n}\left(\hat{\theta}_{n,BB}-\hat{\theta}_n\right)=-\sqrt{n}M_{\theta_0}^{-1}\sum\limits_{i=1}^n \left(w_{in}-\frac{1}{n}\right)m(O_i;\theta_0,h_0) + o_{P_{OW}}(1)
\]
\end{enumerate}
\end{theorem}
Introducing $\hat{\theta}_n$ enables us to express the $\sqrt{n} (\hat{\theta}_{n,BB}-\hat{\theta}_n )$ in terms of the empirical process $\sqrt{n}\left(\mathbb{P}_n^{w}-\mathbb{P}_n\right)$, and hence permits us to swiftly identify the asymptotic posterior distribution. In fact, this convenient expression leads us to the conclude that the conditional distribution of 
$\sqrt{n}(\hat{\theta}_{n,BB}-\hat{\theta}_n )$ given $ O_1,\ldots,O_n$
is asymptotically normal with zero mean and variance $\Sigma$. For large $n$, the posterior mean, for a given realization $O_1,\ldots,O_n$ is  approximately $\hat{\theta}_n$, which depends on $\hat{h}$. The asymptotic posterior variance $\Sigma$ does not depend on the way $h_0$ is estimated and does not automatically match the frequentist asymptotic variance $V$ of $\sqrt{n} (\hat{\theta}_n-\theta_0 )$. A Bayesian/frequentist duality similar to the one in Theorem \ref{thmLLB} does not necessarily hold. Of course, such duality holds if $h_0$ is assumed to be known. Although the expressions of the variances might look complex and comparing them may be difficult, referring to \citet{Henmi/Eguchi:2004} undoubtedly eases the task. To compare these two variances, we interpret $\Sigma$ as a frequentist variance. Specifically, $\Sigma$ is the frequentist asymptotic variance of the estimator of $\theta_0$ with a known nuisance parameter $h_0$, whereas $V$ is the asymptotic variance of the frequentist estimator of $\theta_0$ with an estimated $\hat{h}$ in a first step. \citet{Henmi/Eguchi:2004} allows us to infer that $\Sigma \geq V$ as long as  $\hat{h}$ is acquired using an estimating equation that does not depend on $\theta$. In this case, the conventional two-step procedure induces an asymptotically normal posterior distribution, with the correct mean but with a larger variance. Bayesian credible intervals will have higher than nominal coverage rate unless $h_0$ is assumed to be known. 
\subsection{Linked Bayesian Bootstrap: Bayesian and Frequentist Duality}{\label{LinkedBayesianBootstrap}}
We have seen that a two-step approach, the first of which consists of estimating $h_0$ by $\hat{h}$, and using $\hat{h}$, leads to an inflation of the posterior variance of the targeted parameter if the Bayesian analysis is to be carried out conditional of the value of $\hat{h}$. In the upcoming parts of the paper, we propose a solution to the problem encountered, and argue that such approaches are fully Bayesian. Since we are restricting ourselves to parametric estimation of $h_0$, it is reasonable to assume that $h_0$ satisfies 
$\mathbb{E}u(O;h_0)=0$,
for some function $u : \mathcal{O} \times \mathcal{H} \to \mathbb{R}^q$. In the illustrative example of subsection \ref{AnIllustrativeModel}, we can immediately see that $h_0$ is the solution to $\mathbb{E}u(O;h)=0$, where 
\[
u(O;h)=X^\top\left\{Z-e(X;h)\right\}.
\] 
In a frequentist framework, if $\hat{h}_n$ satisfies 
\[
\sum\limits_{i=1}^n u(O_i;\hat{h}_n)=0,
\]
then under regularity conditions on $u$, we obtain that $
\sqrt{n} (\hat{h}_n-h_0 )$ converges in distribution to a random variable with a $\mathcal{N}(0, \Omega)$ distribution, where 
\[
\Omega= U_{h_0}^{-1} \mathbb{E}_{P_0} \left\{u(O;h_0)u(O;h_0)^\top\right\} U_{h_0}^{-\top} \quad \text{ and } \quad U_{h_0} =\mathbb{E}_{P_O}\left\{\frac{\partial u(O;h_0)}{\partial {h} ^\top}\right\}.
\]
We then consider the frequentist estimator $\hat{\theta}_n$ that satisfies:
\[
\sum\limits_{i=1}^n m(O_i,\hat{\theta}_n,\hat{h}_n)=0.
\]
While the following statements concerning the asymptotic properties of $\hat{\theta}_n$ and $\hat{h}_n$ are not required to prove Theorem \ref{ThmLBBConsistency} , we include them here as they will be essential for our subsequent discussions. Using the results of \citet{Henmi/Eguchi:2004}, it transpires that $\sqrt{n}  (\hat{\theta}_n-\theta_0 )$ converges to a random variable with a $\mathcal{N}(0,V)$ distribution, where 
\[
V= \Sigma -M_{\theta_0}^{-1}M_{h_0}  \Omega M_{h_0}^\top M_{\theta_0}^{-\top},
\]
and that $\sqrt{n}  (\hat{h}_n-h_0 )$ and $\sqrt{n}  (\hat{\theta}_n-\theta_0 )$ are asymptotically independent. 

The Linked Bayesian Bootstrap relies on this simple frequentist mechanism to manipulate the uncertainty attached to the nuisance parameter in the Bayesian estimation of the targeted parameter. Certainly, the underlying motivation for deploying this algorithm is more profound and leverages concepts from causal inference, Bayesian non-parametrics, and decision-theory. \citet{Stephens/etal:2022} aims to justify the use of an adequately chosen plug-in estimate of the propensity score as a Bayesian method. In subsection \ref{ASimpleJustificationoffthisBayesianProcedure}, after stating  Theorem \ref{ThmLBBConsistency} and discussing its implications, we provide a motivation of this algorithm that circumvents the notion of a plug-in and relies exclusively on Bayesian non-parametrics and estimating equations. Before undertaking the theoretical statement, we show an instance of  the Linked Bayesian Bootstrap algorithm for the model in subsection \ref{AnIllustrativeModel}. 
   \begin{algo}{Linked Bayesian Bootstrap for the Partially Linear Model}{\label{StephensLinkedBayesianBootstrap}}
    \begin{enumerate}
        \item For each $j \in \left\{1,\ldots,B\right\}$,
        \begin{enumerate}
        \item  Draw random weights $(w^{(j)}_{1n},\ldots,w^{(j)}_{nn}) \sim \text{Dir}(1,\ldots,1)$.
        \item  Obtain $h^{(j)}$ as the solution to 
        \[
        \sum\limits_{i=1}^n w^{(j)}_{in} X_{i}^\top\left\{Z_i-e(X_i;h)\right\}=0
        \] 
        \item  Find $\theta^{(j)}$ as the solution to
               \[
               \sum\limits_{i=1}^n w^{(j)}_{in}\left\{Y_i-\theta Z_i\right\}\left\{Z_i-e(X_i;h^{(j)})\right\}=0 
               \]
        \end{enumerate}
        \item Output $\theta^{(1)},\ldots, \theta^{(B)}$. 
    \end{enumerate}
\end{algo} 
The uncertainty concerning $h^{(j)}$ is handled through the same Dirichlet weights used to obtain $\theta^{(j)}$, which generate a linkage between the nuisance parameter and the targeted parameter. We now state the following theorem, which assesses the asymptotic properties of the above algorithm.
\begin{theorem}{\label{ThmLBBConsistency}}
 Suppose that the following hold : 
\begin{enumerate}
    \item $\sum\limits_{i=1}^n u(O_i;\hat{h}_{n})=0$ and $\sum\limits_{i=1}^n  m(O_i;\hat{\theta}_{n},\hat{h}_n)=0$
    \item $\sum\limits_{i=1}^n w_{in}u(O_i;\hat{h}_{n})=0$ and  $\sum\limits_{i=1}^n w_{in}m(O_i,\hat{\theta}_{n,BB},\hat{h}_{n,BB})=0$.
  \end{enumerate}
 Then, under the regularity assumptions stated in Assumption \ref{ass3} , 
 \begin{enumerate}
\item[(a)]  \[
\sqrt{n}\left(\hat{\theta}_n-\theta_0\right)=\frac{-1}{\sqrt{n}}M_{\theta_0}^{-1}\sum\limits_{i=1}^n  \left\{m(O_i;\theta_0,h_0) -M_{h_0}U_{h_0}^{-1}u(O_i;h_0)\right\}+o_{P_O}(1),
\]
\item[(b)]
\begin{align*}
\sqrt{n}\left( \right.& \hat{\theta}_{n,BB}   -\left. \hat{\theta}_n\right) \\[6pt] 
& =-M_{\theta_0}^{-1}\sqrt{n}\sum\limits_{i=1}^n \left(w_{in}-\frac{1}{n}\right) \left\{m\left(O_i;{\theta}_0, h_0\right)-M_{h_0}U_{h_0}^{-1}u\left(O_i;h_0\right)\right\}+o_{P_{OW}}(1)
\end{align*}
\end{enumerate}
Those results together imply that 
\[
\sqrt{n}(\hat{\theta}_n-\theta_0 ) \qquad \text{ and } \qquad
\sqrt{n} (\hat{\theta}_{n,BB}-\hat{\theta}_n )\mid O_1,\ldots ,O_n
\]
converge in distribution to random variables with $\mathcal{N}(0,V)$ distributions (almost surely).
\end{theorem}
This result justifies the desirable frequentist properties of the Linked Bayesian Bootstrap by establishing the advantageous Bayesian/frequentist duality.
We observe that the posterior variance of $\theta$ using the Linked Bayesian Bootstrap is not greater than than the variance of the posterior had $h_0$ been known throughout. 

\subsection{Approximation Accuracy}

It is worthwhile to study the standardized posterior distribution, that is the conditional distribution of
\[
\sqrt{n}V^{-0.5}( \hat{\theta}_{n,BB}-\hat{\theta}_n )
\]
given $O_1,\ldots,O_n$, as an approximation to the frequentist sampling distribution of the standardized quantity
\[
\sqrt{n} V^{-0.5} ( \hat{\theta}_{n}- {\theta}_0 ),
\]
through Edgeworth expansions. We suppose that both standardized Bayesian and Frequentist estimators can be well approximated by their respective influence functions, as explained by \citet{Newton/Raftery:1994} and \citet{Shao/Dongsheng:1995}.  Assuming that $\theta_0$ is a real-valued parameter, letting $\Phi$ and $\phi$ denote the standard normal distribution and density respectively, and following the same ideas as in \citet{Newton/Raftery:1994}, the frequentist distribution $F_n$, under regularity conditions found in \citet{Hall:1992} and \citet{Shao/Dongsheng:1995}  by 
\[
\sup\limits_{t \in \mathbb{R}} \left\vert F_{n}(t)  - \Phi(t) - \frac{1}{6 \sqrt{n}} \left[\frac{\mathbb{E}_{P_O}\left\{\lambda(O;\theta_0,h_0)^3\right\}}{\left[\mathbb{E}_{P_{O}}\left\{\lambda(O;\theta_0,h_0)^2\right\}\right]^{3/2}}\right](t^2-1)\phi(t)\right\vert =o (n^{-1/2} ),
\]
while the posterior distribution $F_{n,BB}$ satisfies that
\[
\sup\limits_{t \in \mathbb{R}} \left\vert F_{n,BB}\left(t\mid O_1,\ldots O_n\right)  - \Phi(t) - \frac{2}{6 \sqrt{n}} \left[\frac{\mathbb{E}_{P_O}\left\{\lambda(O;\theta_0,h_0)^3\right\}}{\left[\mathbb{E}_{P_{O}}\left\{\lambda(O;\theta_0,h_0)^2\right\}\right]^{3/2}}\right](t^2-1)\phi(t)\right\vert, 
\]
is $o (n^{-1/2} )$ almost surely  \citep[see][]{Haeusler/etal:1991},
where 
\[
\lambda(O;\theta_0,h_0)=m(O;\theta_0,h_0)-M_{h_0}U_{h_0}^{-1}u(O;h_0)
\]
This indicates that $F_{n,BB}$ approximates $F_n$ as well as $\Phi$ approximates $F_n$, up to second order, since by the triangle inequality, we have that
\[
\sup\limits_{t \in \mathbb{R}} \left\vert F_{n,BB}\left(t\mid O_1,\ldots O_n\right)  - F_n(t) - \frac{1}{6 \sqrt{n}} \left[\frac{\mathbb{E}_{P_O}\left\{\lambda(O;\theta_0,h_0)^3\right\}}{\left[\mathbb{E}_{P_{O}}\left\{\lambda(O;\theta_0,h_0)^2\right\}\right]^{3/2}}\right](t^2-1)\phi(t)\right\vert 
\]
is $o(n^{-1/2})$ almost surely. This approximation can be motivated by the discussion of the results of  \citet{Newton/Raftery:1994} by Albert Y. Lo, who introduced the coefficient of asymptotic accuracy, defined at every $t \in \mathbb{R}$ by 
\[
\lim\limits_{n \to \infty} \frac{\sqrt{n}\left\vert  F_{n}(t)-F_{n,BB}(t\mid O_1,\ldots O_n) \right\vert}{\sqrt{n} \left\vert F_n(t)-\Phi(t)\right\vert}.
\]
and argued that it is equal to $1$ (almost surely) for any $t \in \mathbb{R}$, which emphasizes that the posterior  seen as an approximation of $F_n$ ties with the normal approximation to $F_n$.

However, such approximations are not second order accurate, as the coefficients of $n^{-1/2}$ in the Edgeworth expansions of $F_{n,BB}$ and $F_{n}$ are not equal. We refer to  \citet{Shao/Dongsheng:1995} for a more rigorous treatment of Edgeworth expansions, including regularity assumptions needed for such approximations to be valid. 

\subsection{A One-Step Justification of the Linked Bayesian Bootstrap}{\label{ASimpleJustificationoffthisBayesianProcedure}}
The procedure described in subsection \ref{LinkedBayesianBootstrap} yields a posterior distribution for the targeted parameter with good asymptotic properties and credible intervals at the nominal coverage rate. It substantially differs from the method described in subsection \ref{Disadvantagesofnotadjustingfortheuncertainty} in which we do not account for the uncertainty in estimating the nuisance parameter. While justifications validating the use of a plug-in without propagating uncertainty as fully Bayesian methods may still be constructed and argued for, we opt not to consider such grounds due to the undesirable properties of the posterior as reflected by Theorem \ref{ThmEstimatedNuisance}. However, we aim to find a fully Bayesian interpretation of the Linked Bayesian Bootstrap discussed in subsection \ref{LinkedBayesianBootstrap} that is in alignment with subsection \ref{NoNuisanceParameter} and that bypasses the problematic notion of a plug-in estimator. This interpretation is general and leans on elementary ideas from estimating equations and Bayesian non-parametrics. 

Following the same notation as in the previous sections, let 
\[
a(O;\theta,h)=\begin{bmatrix}
   m(O;\theta,h)\\[3pt]
   u(O;h)
\end{bmatrix}
\in \mathbb{R}^{p+q}.
\]
This interpretation treats $h_0$ as a targeted parameter on an equal footing with $\theta_0$. We define $\eta \in \mathbb{R}^{p+q}$ by $\eta =(\theta^\top,h^\top)^\top$, and write $a(O;\eta)$ in lieu of $a(O;\theta,h)$. Note that the solution $\eta_0$ to the moment condition $\mathbb{E} a(O;\eta)=0$ is given by $\eta_0=(\theta_0^\top,h_0^\top)^\top$, and that estimator $\hat{\eta}_n$ satisfies
\[
\sum\limits_{i=1}^n a(O_i,\hat{\eta}_n)=0,
\]
is given by $\hat{\eta}_n=(\hat{\theta}_n^\top,\hat{h}_n^\top)^\top$, where $\hat{h}_n$ and $\hat{\theta}_n$ satisfy respectively the first assumption of Theorem \ref{ThmLBBConsistency}:
\[
\sum\limits_{i=1}^n u(O_i,\hat{h}_n)=0 \text{ and } \sum\limits_{i=1}^n m(O_i,\hat{\theta}_n,\hat{h}_n)=0.
\]
It then becomes clear that 
\[
\sqrt{n}\begin{pmatrix}
    \hat{\theta}_n-\theta_0\\
    \hat{h}_n-h_0
\end{pmatrix}
=\sqrt{n} \left(\hat{\eta}_n-\eta_0\right)
\to \mathcal{N}(0,\Lambda),
\]
where the $(p+q) \times (p+q)$ matrix $\Lambda$ is given by
\[
\Lambda=\left[\mathbb{E}_{P_0}\left\{\frac{\partial a(O;\eta_0)}{\partial \eta^\top }\right\}\right]^{-1}\left[\mathbb{E}_{P_0} \left\{a(O;\eta_0)a(O;\eta_0)^\top\right\}\right]\left[\mathbb{E}_{P_0}\left\{\frac{\partial a(O;\eta_0)}{\partial \eta^\top }\right\}\right]^{-\top}
\]
and can be shown, using the results \citet{Henmi/Eguchi:2004}, to satisfy
\[
\Lambda=\begin{pmatrix}
    V & 0_{p \times q}\\
    0_{q \times p} &\Omega 
\end{pmatrix}.
\]
Moving away from this frequentist framework, and by following the same reasoning as in subsection \ref{NoNuisanceParameter}, we can obtain a posterior for $\eta$ using the Bayesian bootstrap, by solving 
\[
\sum_{i=1}^n w_{in}a(O_i,\hat{\eta}_{n,BB})=0.
\]
It is immediately revealed that $\hat{\eta}_{n,BB}=({\hat{\theta}_{n,BB}}^\top,{\hat{h}_{n,BB}}^\top)^\top$, where  $\hat{h}_{n,BB}$ and $\hat{\theta}_{n,BB}$ satisfy respectively the second assumption of Theorem \ref{ThmLBBConsistency}:
\[
\sum\limits_{i=1}^n w_{in}u(O_i,\hat{h}_{n,BB})=0 \quad \text{ and } \quad \sum\limits_{i=1}^n w_{in} m(O_i,\hat{\theta}_{n,BB},\hat{h}_{n,BB})=0.
\]
We can hence conclude, using subsection \ref{NoNuisanceParameter}, that
\[
\sqrt{n}\begin{pmatrix}
    \hat{\theta}_{n,BB}-\hat{\theta}_n\\
    \hat{h}_{n,BB}-\hat{h}_n
\end{pmatrix} 
\mid O_1,\ldots,O_n
=\sqrt{n} \left(\hat{\eta}_{n,BB}-\hat{\eta}_n\right) \mid O_1,\ldots,O_n
\]
converges to a random variable with a $\mathcal{N} (0,\Lambda)$ distribution.  Thus $\theta$ and $h$ are \textit{a posteriori} asymptotically independent, and by taking the marginal asymptotic posterior of $\theta$, we recover the result of Theorem \ref{ThmLBBConsistency}. The proof (presented in the supplementary material) relies on the notion of a plugged-in nuisance parameter to infer about the marginal posterior distribution of $\theta$ without taking into account the joint posterior distribution. However, the above process first establishes the properties of the joint posterior distribution which are then used to obtain the marginal posterior distribution of the targeted parameter. The Linked Bayesian Bootstrap could hence be understood as a one-step procedure, yielding joint posterior inference concerning $h$ and $\theta$. The two-step approach is a simple computational strategy to solving the system of equations represented by the augmented function $a(O;\eta)$ and attaches little importance to the controversial notions of using a plug-in estimator and of the uncertainty propagation of the propensity model in the posterior inference for $\theta$. This view of the Linked Bayesian Bootstrap overcomes the difficulty in combining the two steps into one that jointly estimates propensity scores and causal effects without feedback; see \citet{Zigler:2016}, and results from a use of a non-parametric Bayes theorem. 

Formulating the problem in terms of estimating equations instead of two optimizations, as in Algorithm \ref{LBBOpt}, is of crucial importance. Indeed, the two optimizations may not be formulated as a single joint optimization problem for which the Loss-Likelihood Bootstrap can be applied. To illustrate this fact, let us assume for simplicity that both $\theta$ and $h$ are one-dimensional parameters and that a bivariate real-valued function $l(O;\theta,h)$ has its gradient equal to $\left[m\left(O;\theta,h\right),u\left(O;h\right)\right]$ or to $\left[u\left(O;h\right),m\left(O;\theta,h\right)\right]$. Then, using the fact that the gradient of a function is a conservative vector field, we obtain that
\[
\frac{\partial m\left(O;\theta,h\right)}{\partial \theta}=\frac{\partial u(O;h)}{ \partial h} \qquad\text{ or } \qquad \frac{\partial m\left(O;\theta,h\right)}{\partial h}=\frac{\partial u(O;h)}{ \partial \theta}=0. 
\] 
Both of these conditions can be shown to fail if we use, for example, the IPW estimator o subsection \ref{WeightingEstimators}.  The two estimating equations however, can be combined into a system which identifies the parameters $\theta_0$ and $h_0$.

\section{Simulations}{\label{Simulations}}
\subsection{G-Estimation}
We consider the model in subsection \ref{AnIllustrativeModel} and implement Algorithm \ref{StephensLinkedBayesianBootstrap}. The data generating mechanism is as follows: $X \in \mathbb{R}^6$ be distributed according to a multivariate Normal distribution with zero mean and covariance matrix $\Sigma_X$, whose $(i,j)$-th entry equals $0.8^{\vert i-j \vert}$. We then generate the binary treatment $Z$ and outcome $Y$ such that 
\begin{align}\label{gmod1}
    Y &= \theta_0 Z + g_0(X)  +U \\ 
    Z|X& \sim \text{Bern}(e_0(X)),
\end{align}
where $\theta_0=3$, $U \sim \mathcal{N}(0,1)$
\begin{align*}
    g_0(X)&=X_1+\exp(X_2-1)+\vert X_3 \vert + \exp(X_4-1)+\vert X_5 \vert +\vert X_5 X_6 \vert\\
    e_0(X)&= e(X;h_0)=\expit(h_0^\top X)
\end{align*}
with $h_0=(1,2,3,4,5,6)^\top/16$. 
We run a simulation across $1000$ replicates for different sample sizes, where at each iteration we have estimated the frequentist ATE using the moment restriction $\mathbb{E}m(O;\theta_0,h_0)=0$, where
\[
m(O;\theta,h)=(Y-\theta Z)\left\{Z-e(X;h)\right\}.
\]
We estimate $h_0$ by $\hat{h}$ obtained through logistic regression with intercept. For each replicate, we implement Algorithm \ref{StephensLinkedBayesianBootstrap} and derive the posterior distribution for the targeted parameter using 1000 Bayesian bootstrap samples. We report the results in Table \ref{ATErmnPost}, which underscores the Bayesian/frequentist duality proved in subsection \ref{LinkedBayesianBootstrap}. 
\begin{table}[h]
\centering
\begin{tabular}{@{}lrrrr}
\hline
& \multicolumn{3}{c}{$n$} \\[3pt]
& \multicolumn{1}{c}{$250$}
& \multicolumn{1}{c}{$500$} & \multicolumn{1}{c}{$1000$}
\\
\hline
Average of Posterior Means    & $3.01$ & $2.99$  & $3.00$  \\
Empirical Frequentist Mean  &   $3.01$      &  $2.99$        & $3.00$       \\
Average of Posterior Variances $(\times n)$   & $20.89$  & $21.44$  & $21.71$   \\
Empirical Frequentist Variance $(\times n$)&   $23.19$       & $22.27$ &     $21.57$             \\
Average Sandwich Estimate   & $21.99$  & $22.06$  & $21.98$ \\ 
\hline
Average Bayesian credible interval & $(2.45,3.57)$  &$(2.59,3.39)$ & $(2.71,3.29)$\\
Frequentist confidence interval &$(2.39,3.58)$ &$(2.58,3.42)$ &$(2.71,3.29)$\\
Posterior Coverage    & $92.6$  & $94.0$  & $94.2$ \\
\hline
\end{tabular}
\caption{Comparison of the frequentist and posterior distributions of the ATE estimated through the score $m$ and the Linked Bayesian Bootstrap. The true values of the ATE and the asymptotic variance of the sandwich variance are respectively $3$ and $22.32$\label{ATErmnPost}
}
\end{table}
 Before supplying other causal quantities of interest for which the 'Linked Bayesian Bootstrap' can be applied, we succinctly illustrate the impact of using a fixed estimate, without accounting for its uncertainty, of the propensity score in the Bayesian bootstrap procedure. The results are reported in Table \ref{ATErmnfixedh}. We see that for $n=1000$, the empirical posterior variances of the ATE are identical, as the theory suggests. As regards coverage, the coverage rate is at the nominal level only when $h_0$ is known. In the remaining cases, the posterior manifests coverage at a higher rate than the nominal level as implied by the analysis following Theorem \ref{ThmEstimatedNuisance}.
 
\begin{table}[t]
\centering
\begin{tabular}{@{}lrrrrrrr}
\hline
& \multicolumn{6}{c}{$n$} \\[3pt]
& \multicolumn{2}{c}{$250$}
& \multicolumn{2}{c}{$500$} & \multicolumn{2}{c}{$1000$}\\
\hline
 $h_0$   & 87.33 & (94.5) & 85.38 & (93.2)  & 83.39 & (94.5) \\
Logist. Regr. (with intercept) &  86.55 & (99.8)&  84.51 & (99.9)& 82.93 & (100.0)    \\
Logist. Regr. (no intercept)   & 89.57 & (96.4) & 86.61 & (96.0) & 83.58& (97.1)  \\
LLB (with intercept)&   88.46 & (99.8)      & 85.35& (100.0)&     83.24 & (100.0)          \\
LLB (no intercept)  & 91.32 & (96.5) & 87.41 & (95.9) & 83.94 & (97.0)\\ 
\hline
\end{tabular}
\caption{Averages of the empirical posterior variances (times $n$) of the treatment effect when a fixed plug-in estimator of $h_0$ is used, as in Algorithm \ref{algonuisancefixed}. In brackets are the posterior coverage rates based on the $2.5\%$ and $97.5\%$ percentiles. In the first row, the title $h_0$ indicates that the true value of $h_0$ is used, whereas 'Logist. Regr.' indicates that a frequentist logistic regression of $Z$ against $X$ is used. 'LLB' indicates that the loss-likelihood bootstrap is used, with an objective function equal to the log-likelihood. The posteriors were generated using $500$ Bayesian bootstrap samples and then averaged.\label{ATErmnfixedh}}
\end{table}

\subsection{Weighting Estimators}{\label{WeightingEstimators}}
We now perform a Bayesian analysis of the ATE estimated using the IPW estimator derived from the estimating equation \eqref{ipwee}. We generate confounders $X_1 \sim \mathcal{N}(1/2,1/{16})$ and $X_2 \sim \mathcal{N}(1,1/4)$, and let $X=(X_1,X_2)^\top$. We generate the binary treatment $Z$ and the outcome $Y$ according to the following structural model \eqref{gmod1}
where $\theta_0=3$, $U \sim \mathcal{N}(0,1)$, $g_0(X)=2X_1+X_2$, $e_0(X)= e(X;h_0)=\expit(h_0^\top X)$, and $h_0^\top=(1,1/2)$. As in the previous study, the experiment is run across 1000 replicates for three different sample sizes. At each iteration, we estimate the frequentist value of the ATE using the IPW estimator using an estimate of $h_0$ by a frequentist logistic regression. 
The IPW estimator does not depend on an outcome model but its inferential validity depends exclusively on the correct specification of the propensity score and on the fact that the propensity score does not take the limiting values $0$ and $1$.

We implement the Linked Bayesian Bootstrap algorithm at each iteration using $1000$ Bayesian bootstrap samples.The results are reported in Table \ref{ATEIPWPost}, which highlight the Bayesian/frequentist duality proved in Theorem \ref{ThmLBBConsistency}. In order to assess the behavior of the posterior distribution of the ATE estimated using the IPW in a conventional plug-in approach, in which a fixed estimate $h_0$ is used, we mirror the experiments run in the previous study. The results are reported in Table \ref{ATEIPWfixedh} and are in line with the findings of Theorem \ref{ThmEstimatedNuisance}. In particular, the posterior variance is not affected by the fixed estimate of $h_0$ and is higher than the one obtained under the Linked Bayesian Bootstrap. Moreover, coverage at the nominal level in these conventional two-step approaches is only achieved when $h_0$ is known, as predicted by the theory. Weighting estimators can also be used to estimate the ATT $\rho_0$, which satisfies  $\mathbb{E}r(O;\rho_0,h_0)=0$, where
\[
r(O;\rho,h)= ZY-\frac{Y(1-Z)e(X;h)}{1-e(X;h)}-\rho Z.
\] 
We carry-out a simulation study to perform a Bayesian analysis on this estimator. The results are reported in Table \ref{ATTPost} for the Linked Bayesian Bootstrap, and in Table \ref{ATTfixedh} for the conventional two-step approach. In this model, the ATT and ATE are equal.

\begin{table}[t]
\centering
\begin{tabular}{@{}lrrrr}
\hline
& \multicolumn{3}{c}{$n$} \\[3pt]
& \multicolumn{1}{c}{$250$}
& \multicolumn{1}{c}{$500$} & \multicolumn{1}{c}{$1000$}
\\
\hline
Average of Posterior Means    & $3.00$ & $3.00$  & $3.00$  \\
Empirical Frequentist Mean  &   $3.00$      &  $3.00$        & $3.00$       \\
Average of Posterior Variances $(\times n)$   & $6.46$  & $6.01$  & $5.81$   \\
Empirical Frequentist Variance $(\times n$)&   $6.04$       & $5.62$ &     $5.82$             \\
Average Sandwich Estimate   & $6.24$  & $5.97$  & $5.82$ \\ 
\hline
Average Bayesian credible interval & $(2.68,3.31)$  &$(2.78,3.21)$ & $(2.85,3.15)$\\
Frequentist confidence interval &$(2.70,3.30)$ &$(2.80,3.20)$ &$(2.86,3.15)$\\
Posterior Coverage    & $95.4$  & $95.3$  & $94.8$ \\
\hline
\end{tabular}
\caption{Comparison of the frequentist and posterior distributions of the ATE estimated using the IPW estimator and the Linked Bayesian Bootstrap respectively. The true values of the ATE and the asymptotic variance of the sandwich variance are respectively $3$ and $5.68$.\label{ATEIPWPost}}
\end{table}

\begin{table}[t]
\centering
\begin{tabular}{@{}lrrrrrrr}
\hline
& \multicolumn{6}{c}{$n$} \\[3pt]
& \multicolumn{2}{c}{$250$}
& \multicolumn{2}{c}{$500$} & \multicolumn{2}{c}{$1000$}\\
\hline
 $h_0$   & 51.04 & (95.5) & 51.06 & (94.3)  & 51.44 & (93.7) \\
Logist. Regr. (with intercept) &  52.34 & (100.0)&  51.71 & (100.0)& 51.75 & (100.0)    \\
Logist. Regr. (no intercept)   & 51.78 & (100.0) & 51.58 & (100.0) & 51.67& (100.0)  \\
LLB (with intercept)&   53.32 & (100.0)      & 52.15& (100.0)&     51.95 & (100.0)          \\
LLB (no intercept)  & 52.41 & (100.0) & 51.89 & (100.0) & 51.81 & (100.0)\\ 
\hline
\end{tabular}
\caption{Averages of the empirical posterior variances (times $n$) of the average treatment effect estimated using the IPW estimator when a fixed plug-in estimator of $h_0$ is used. In brackets are the posterior coverage rates based on the $2.5\%$ and $97.5\%$ percentiles. The methods used for obtaining the estimate $\hat{h}$ of $h_0$ are the same as those in Table \ref{ATErmnfixedh}. In the first row, the title $h_0$ indicates that the true value of $h_0$ is used, whereas 'Logist. Regr.' indicates that a frequentist logistic regression of $Z$ against $X$ is used. 'LLB' indicates that the loss-likelihood bootstrap is used, with an objective function equal to the log-likelihood. The posteriors were generated using $500$ Bayesian bootstrap samples and then averaged. \label{ATEIPWfixedh}}
\end{table}

\begin{table}[h]
\centering
\begin{tabular}{@{}lrrrr}
\hline
& \multicolumn{3}{c}{$n$} \\[3pt]
& \multicolumn{1}{c}{$250$}
& \multicolumn{1}{c}{$500$} & \multicolumn{1}{c}{$1000$}
\\
\hline
Average of Posterior Means    & $3.00$ & $3.00$  & $3.00$  \\
Empirical Frequentist Mean  &   $3.00$      &  $3.00$        & $3.00$       \\
Average of Posterior Variances $(\times n)$   & $7.55$  & $6.75$  & $6.51$   \\
Empirical Frequentist Variance $(\times n$)&   $7.02$       & $6.36$ &     $6.49$             \\
Average Sandwich Estimate   & $7.19$  & $6.65$  & $6.49$ \\ 
\hline
Average Bayesian credible interval & $(2.66,3.32)$  &$(2.76,3.22)$ & $(2.84,3.16)$\\
Frequentist confidence interval &$(2.68,3.31)$ &$(2.77,3.22)$ &$(2.83,3.15)$\\
Posterior Coverage    & $95.4$  & $95.1$  & $95.0$ \\
\hline
\end{tabular}
\caption{Comparison of the frequentist and posterior distributions of the ATT, estimated using the score $r$ under the Linked Bayesian Bootstrap. The true values of the ATT and the asymptotic variance of the sandwich variance are respectively $3$ and $6.18$.\label{ATTPost}}
\end{table}

\begin{table}[t]
\centering
\begin{tabular}{@{}lrrrrrrr}
\hline
& \multicolumn{6}{c}{$n$} \\[3pt]
& \multicolumn{2}{c}{$250$}
& \multicolumn{2}{c}{$500$} & \multicolumn{2}{c}{$1000$}\\
\hline
 $h_0$   & 37.82 & (94.7) & 37.24 & (95.1)  & 36.31 & (95.2) \\
Logist. Regr. (with intercept) &  38.55 & (100.0)&  37.20 & (100.0)& 36.54 & (100.0)    \\
Logist. Regr. (no intercept)   & 37.78 & (100.0) & 37.09 & (100.0) & 36.50& (100.0)  \\
LLB (with intercept)&   40.25 & (100.0)      & 37.96& (100.0)&     36.90 & (100.0)          \\
LLB (no intercept)  & 38.87 & (100.0) & 37.63 & (100.0) & 36.75 & (100.0)\\ 
\hline
\end{tabular}
\caption{Averages of the empirical posterior variances (times $n$) of the average treatment effect on the treated when a fixed plug-in estimator of $h_0$ is used. In brackets are the posterior coverage rates based on the $2.5\%$ and $97.5\%$ percentiles. In the first row, the title $h_0$ indicates that the true value of $h_0$ is used, whereas 'Logist. Regr.' indicates that a frequentist logistic regression of $Z$ against $X$ is used. 'LLB' indicates that the loss-likelihood bootstrap is used, with an objective function equal to the log-likelihood. The posteriors were generated using $500$ Bayesian bootstrap samples and then averaged.\label{ATTfixedh}}
\end{table}

\section{Discussion}

We have derived the properties of Bayesian posterior distributions defined by targeting parameters using estimating equations and computed using the Bayesian bootstrap.  To do this, we have exploited Bayesian/frequentist duality results to justify the use of specific computational strategies such as the Linked Bayesian Bootstrap.  

In the causal setting, when the form of the outcome model is unknown, adjustments based on the propensity score are valuable and can lead to consistent estimation of the causal quantities of interest as suggested by the theory of estimating equations. Certainly, we are limited to the knowledge of the treatment allocation process up to parametric form. These assumptions have lied at the heart of the frequentist parametric modeling in causal inference and have led to developments in singly robust models, in which the functional form of the outcome model is not known. In the Bayesian framework, a joint estimation of the propensity score and the causal quantity of interest can lead to an unwanted behavior of the Bayesian estimates when the outcome model is mis-specified. Such feedback, from the mis-specified outcome model to the propensity score model, halts the propensity score from playing its balancing role and results in biased estimation of the treatment effect. Remedies for the issue of feedback have been proposed such as the cutting feedback and two-step approaches that we have recapped briefly in subsection \ref{ApproachestoBCI}. The 'Linked Bayesian Bootstrap' provided another view of two-step methods using the Bayesian bootstrap approximation to the Dirichlet process posterior. The term 'Linked' is to emphasize that the uncertainty attached to the propensity score should be propagated using the same Dirichlet weights that are used to obtain Bayesian estimates of the treatment effect. In this paper, we have studied the Linked Bayesian Bootstrap in detail, establishing its frequentist desirable properties including the Bayesian/frequentist duality it satisfies. This particular result is an analogue of the parametric Bernstein-von Mises theorem and enables us to obtain Bayesian credible intervals at the correct nominal levels. We have additionally interpreted the Linked Bayesian Bootstrap as a fully Bayesian one-step procedure leading to a joint posterior inference about the treatment and the propensity score, with no feedback. This interpretation, derived in subsection \ref{ASimpleJustificationoffthisBayesianProcedure}, avoided the problematic of propagating the uncertainty of the propensity score model and the problematic notion of a plug-in and treated both parameters as parameters of interest. Finally, we are conscious of the modeling restrictions the parametric assumptions impose. Certainly, we are aware that highly flexible methods can be deployed to estimate causal effects, as shown in \citet{Antonelli/etal:2020}, \citet{Hahn/etal:2020}, and \citet{Vegetabile/etal:2020}. Such methods are not the focus of this paper and constitute the topic of ongoing work.


\textbf{Acknowledgements:} The authors acknowledge the support of the Natural Sciences and Engineering Research Council of Canada (NSERC) and the  Institut des Sciences Math\'ematiques (ISM).



\vspace{0.4in}

\appendix

\begin{center}
\textsc{Supplementary Material}    
\end{center}
\begin{assumption}[Technical Assumptions]{\label{ass0}}
The following assumptions \citep[see][]{Cheng/Huang:2010} specify the relevant probability spaces required in the theoretical statements.  We assume
\begin{enumerate}
    \item The data are realizations from the from probability space $(\mathcal{O},B,P_O)$.
    \item The Bayesian bootstrap weights are defined on probability space $(\mathcal{W},\mathcal{C},P_W)$.
    \item Observation $O_i$ is the $i$-th coordinate of the canonical projection from $(\mathcal{O}^\infty,B^\infty,P_O ^\infty)$.
    \item For the joint randomness, coming from the observed data and from the weights, is defined on the product probability space
    \[
    (\mathcal{O}^\infty,B^\infty,P_O ^\infty) \times (\mathcal{W},\mathcal{C},P_W)=(\mathcal{O}^\infty \times \mathcal{W},B^\infty \times \mathcal{C},P_O ^\infty \times P_W).
    \]
    The joint probability measure is hence $P_{OW}=P_{O}^\infty \times P_W$. We write $P_O$ in lieu of $P_O^{\infty}$ for simplicity when necessary. 
\end{enumerate}
    
\end{assumption}
\begin{assumption}[Assumptions for Theorem \ref{ThmEstimatedNuisance}]{\label{ass2}}
The following assumptions are not difficult to verify in practice. They are usually verified by imposing conditions on $m$; see \citet{Tsiatis:2006}. Those assumptions are to ensure 
\[
\frac{1}{n}\sum\limits_{i=1}^n \frac{\partial m(O_i;\tilde{\theta},\tilde{h})}{\partial \theta ^\top} \text{ and } \sum\limits_{i=1}^n w_{in} \frac{\partial m(O_i;\tilde{\theta},\tilde{h})}{\partial \theta ^\top}
\]
converge in probability to 
\[
{M_{\theta_0}}=\mathbb{E}_{P_O} \left\{ \left. \frac{\partial m\left(O;\theta,h_0 \right)}{\partial {\theta}^\top} \right |_{\theta = \theta_0} \right\},
\] 
provided that  $\tilde{\theta}$ and $\tilde{h}$ are consistent estimators of $\theta_0$ and $h_0$ respectively. We refer to \citet{Praestgaard/Wellner:1993} as well as \citet{VanderVaart/Wellner:1996} for an in-depth treatment concerning the exchangeable bootstrap for Glivenko-Cantelli and Donsker classes.
    \begin{enumerate}
        \item The class 
        $
        \displaystyle{\left\{\frac{\partial m(O;\theta,{h})}{\partial \theta ^\top}, \Vert \theta -\theta_0\Vert_{p,2}<\delta_1, \Vert h -h_0 \Vert_{q,2} <\delta_2 \right\}}
        $
        is $P_O$-Glivenko-Cantelli for some $\delta_1,\delta_2>0$, and the function
        $(\theta,h) \displaystyle{\mapsto\mathbb{E}_{P_O} \left\{\frac{\partial m(O;\theta,h)}{\partial \theta}\right\}}$
        is continuous.
        
        \item  The class is $\displaystyle{\left\{\frac{\partial m(O;\theta_0,{h})}{\partial h ^\top}, \Vert h -h_0 \Vert_{q,2} <\delta \right\}}$ is $P_O$- Glivenko-Cantelli for some $\delta>0$,  and the function
        $
        h \mapsto\displaystyle{\mathbb{E}_{P_O} \left\{\frac{\partial m(O;\theta_0,h)}{\partial h}\right\}}
        $
        is continuous.
        
        \item  $\displaystyle{M_{\theta_0}}$ is an invertible $p \times p$ matrix.
    \end{enumerate}
\end{assumption}
\begin{proof}[Proof of Theorem \ref{ThmEstimatedNuisance}]{\label{Proof1}}
The main idea behind the proof is to establish the unconditional properties of $\sqrt{n} (\hat{\theta}_{n}-\theta_0 )$ and $\sqrt{n} (\hat{\theta}_{n,BB}-{\theta}_0 )$. We then subtract these two expressions to obtain an expression of $\sqrt{n} (\hat{\theta}_{n,BB}-\hat{\theta}_n )$ in terms of the empirical process $\sqrt{n} \left(\mathbb{P}_n^{w}-\mathbb{P}_n \right)$ to assess its conditional and unconditional behavior.
   The proof of the fact that
   \[
\sqrt{n}\left(\hat{\theta}_n-\theta_0\right)=\frac{-1}{\sqrt{n}}M_{\theta_0}^{-1}\sum\limits_{i=1}^n  \left\{m(O_i;\theta_0,h_0) + M_{h_0}\psi(O_i;h_0)\right\}+o_{P_O}(1),
\]
   can be found in \citet{Newey/McFadden:1994}.
   In order to prove the second the statement, we first show that 
   \[
     \sqrt{n}\left(\hat{\theta}_{n,BB}-\theta_0\right)=-M_{\theta_0}^{-1}\left\{\sqrt{n}\sum\limits_{i=1}^n w_{in} m\left(O_i;{\theta}_0, h_0\right)+\frac{M_{h_0}}{\sqrt n} \sum\limits_{i=1}^n \psi(O_i;h_0)\right\}+o_{P_{OW}}(1)
   \]
   We first start by expanding the  term below around $\theta_0$:
\[
\sum\limits_{i=1}^n w_{in} m\left(O_i;\hat{\theta}_{n,BB}, \hat{h}\right)=
\sum\limits_{i=1}^n w_{in} m\left(O_i;{\theta}_0, \hat{h}\right) + \left\{\sum\limits_{i=1}^n w_{in}\frac{\partial m(O_i;\tilde{\theta},\hat{h})}{\partial \theta ^\top} \right\}\left(\hat{\theta}_{n,BB}-\theta_0\right),
\]
for some $\tilde{\theta}$ on the line segment joining $\hat{\theta}_{n,BB}$ to $\theta_0$.
We now examine 
\begin{align*}
&\sqrt{n}\sum\limits_{i=1}^n w_{in} m\left(O_i;{\theta}_0, \hat{h}\right)\\
&= \sqrt{n}\sum\limits_{i=1}^n w_{in} m\left(O_i;{\theta}_0, h_0\right)+\sqrt{n}\left\{\sum\limits_{i=1}^n w_{in} \frac{\partial m(O_i;\theta_0,\tilde{h})}{\partial h ^\top} \right\}\left(\hat{h}-h_0\right)\\
&=\sqrt{n}\sum\limits_{i=1}^n w_{in} m\left(O_i;{\theta}_0, h_0\right)+\left\{\sum\limits_{i=1}^n w_{in} \frac{\partial m(O_i;\theta_0,\tilde{h})}{\partial h ^\top} \right\}\left\{\frac{1}{\sqrt n} \sum\limits_{i=1}^n \psi\left(O_i;h_0\right) +o_{P_{O}}(1) \right\}
\end{align*}
   for some $\tilde{h}$ on the line segment joining $\hat{h}$ to $h_0$. 
Using Assumption \ref{ass2}, we have that 
\[
\sum\limits_{i=1}^n w_{in} \frac{\partial m(O_i;\theta_0,\tilde{h})}{{\partial h ^\top}} = M_{h_0}+o_{P_{OW}}(1) 
\hspace{5mm}
\text{and}
\hspace{5mm}
\sum\limits_{i=1}^n w_{in}\frac{\partial m(O_i;\tilde{\theta},\hat{h})}{\partial \theta ^\top}  = M_{\theta_0}+o_{P_{OW}}(1).
\]
Moreover, since 
\[
\frac{1}{\sqrt n} \sum\limits_{i=1}^n \psi(O_i;h_0)=O_{P_{OW}}(1),
\]
by the central limit theorem together with the fact that it does not depend on the bootstrap weights,
we get that
\[
\sqrt{n}\sum\limits_{i=1}^n w_{in} m\left(O_i;{\theta}_0, \hat{h}\right)=\sqrt{n}\sum\limits_{i=1}^n w_{in} m\left(O_i;{\theta}_0, h_0\right)+M_{h_0}\left\{\frac{1}{\sqrt n} \sum\limits_{i=1}^n \psi(O_i;h_0)\right\}+o_{P_{OW}}(1).
\]
In addition, 
\[
\sqrt{n}\sum\limits_{i=1}^n w_{in} m\left(O_i;{\theta}_0, h_0\right)=O_{P_{OW}}(1).
\]
Combining all the results together, we get that
\begin{align*}
&\sqrt{n}\left(\hat{\theta}_{n,BB}-\theta_0\right)\\
&=-\left\{M_{\theta_0}^{-1} +o_{P_{OW}}(1)\right\} \left\{o_{P_{OW}}(1)+\sqrt{n}\sum\limits_{i=1}^n w_{in} m\left(O_i;{\theta}_0, h_0\right)
+\frac{M_{h_0}}{\sqrt n} \sum\limits_{i=1}^n \psi(O_i;h_0)+o_{P_{OW}}(1)\right\}\\
&=-M_{\theta_0}^{-1}\left\{\sqrt{n}\sum\limits_{i=1}^n w_{in} m\left(O_i;{\theta}_0, h_0\right)+\frac{M_{h_0}}{\sqrt n} \sum\limits_{i=1}^n \psi(O_i;h_0)\right\}+o_{P_{OW}}(1).
\end{align*}
By writing $\sqrt{n}\left(\hat{\theta}_{n,BB}-\hat{\theta}_n\right)=\sqrt{n}\left(\hat{\theta}_{n,BB}-{\theta}_0\right)-\sqrt{n}\left(\hat{\theta}_n-\theta_0\right)$, we obtain
\[
\sqrt{n}\left(\hat{\theta}_{n,BB}-\hat{\theta}_n\right)=-\sqrt{n}M_{\theta_0}^{-1}\sum\limits_{i=1}^n \left(w_{in}-\frac{1}{n}\right)m(O_i;\theta_0,h_0) + o_{P_{OW}}(1).
\]
\end{proof}
\begin{assumption}[Assumption for Theorem \ref{ThmLBBConsistency}]{\label{ass3}} The following assumptions are similar to those in Assumption \ref{ass2}.
\begin{enumerate}
    \item The class 
    $
    \displaystyle{\left\{\frac{\partial m(O;\theta,{h})}{\partial \theta ^\top}, \Vert \theta -\theta_0\Vert_{p,2}<\delta_1, \Vert h -h_0 \Vert_{q,2} <\delta_2 \right\}}
    $
    is $P_O$-Glivenko-Cantelli for some $\delta_1,\delta_2>0$, and the function
        $
        (\theta,h) \mapsto\displaystyle{\mathbb{E}_{P_O} \left\{\frac{\partial m(O;\theta,h)}{\partial \theta}\right\}}
        $
        is continuous.
    \item  The class is $\displaystyle{\left\{\frac{\partial m(O;\theta_0,{h})}{\partial h ^\top}, \Vert h -h_0 \Vert_{q,2} <\delta \right\}}$ is $P_O$- Glivenko-Cantelli for some $\delta>0$, and the function
        $
        h \mapsto\displaystyle{\mathbb{E}_{P_O} \left\{\frac{\partial m(O;\theta_0,h)}{\partial h}\right\}}
        $
        is continuous.
        
    \item The class 
    $
    \displaystyle{\left\{\frac{\partial u(O;{h})}{\partial h ^\top}, \Vert h -h_0 \Vert_{q,2} <\delta \right\}}
    $
    is $P_O$- Glivenko-Cantelli for some $\delta>0$, and the function
        $
        h \mapsto\displaystyle{\mathbb{E}_{P_O} \left\{\frac{\partial u(O;h)}{\partial h}\right\}}
        $
        is continuous.
        
    \item  $\displaystyle{M_{\theta_0}}$ is an invertible $p \times p$ matrix.
    \item  $\displaystyle{U_{h_0}}$ is an invertible $q \times q$ matrix.
    \item $\hat{h}_n$ and $\hat{h}_{n,BB}$ are unconditionally consistent estimators of $h_0$.
    \item $\hat{\theta}_n$ and $\hat{\theta}_{n,BB}$ are unconditionally consistent estimators of $\theta_0$.
\end{enumerate}
\end{assumption}

\begin{proof}[Proof of Theorem \ref{ThmLBBConsistency}]{\label{ProofLBB}}
First, we note that
\[
\sqrt{n} \left(\hat{h}_{n}-h_0\right) =-U_{h_0}^{-1} \sqrt{n}\sum\limits_{i=1}^n u\left(O_i;h_0\right)+o_{P_{O}}(1).
\]
By substituting the above expression into the first part of Theorem \ref{ThmEstimatedNuisance},
we obtain that
\[
\sqrt{n}\left(\hat{\theta}_n-\theta_0\right)=\frac{-1}{\sqrt{n}}M_{\theta_0}^{-1}\sum\limits_{i=1}^n  \left\{m(O_i;\theta_0,h_0) -M_{h_0}U_{h_0}^{-1}u(O_i;h_0)\right\}+o_{P_O}(1).
\]
We now proceed in a similar manner as in the Proof \ref{Proof1}, to study $\sqrt{n}\left(\hat{\theta}_{n,BB}-\theta_0\right)$. We observe that
\[
\sqrt{n} \left(\hat{h}_{n,BB}-h_0\right) =-U_{h_0}^{-1} \sqrt{n}\sum\limits_{i=1}^n w_{in}u\left(O_i;h_0\right)+o_{P_{OW}}(1).
\]
By a similar expansion, we hence obtain 
\begin{align*}
\sum\limits_{i=1}^n w_{in}m\left(O_i;\hat{\theta}_{n,BB},\hat{h}_{n,BB}\right)  &=\sum\limits_{i=1}^n w_{in}m\left(O_i;\theta_0,\hat{h}_{n,BB}\right) \\[5pt]
& \quad +\left\{\sum\limits_{i=1}^n w_{in}\frac{\partial m\left(O_i;\tilde{\theta},\hat{h}_{n,BB}\right)}{\partial \theta ^\top}\right\} \left(\hat{\theta}_{n,BB}-\theta_0 \right)
\end{align*}
\begin{align*}
\sum\limits_{i=1}^n w_{in}m\left(O_i;\theta_0,\hat{h}_{n,BB}\right)&=\sum\limits_{i=1}^n w_{in}m\left(O_i;\theta_0,h_0\right) +\left\{\sum\limits_{i=1}^n w_{in} \frac{\partial m(O_i;\theta_0,\tilde{h})}{\partial h ^\top}\right\} \left(\hat{h}_{n,BB}-h_0 \right)\\
&=\sum\limits_{i=1}^n w_{in}m(O_i;\theta_0,h_0)-M_{h_0}U_{h_0}^{-1}  \sum\limits_{i=1}^n w_{in}u(O_i;h_0) +o_{P_{OW}}(1)
\end{align*}
and therefore,
\begin{align*}
    \sqrt{n}\left(\hat{\theta}_{n,BB}-\theta_0 \right)
    &=-M_{\theta_0}^{-1}\sqrt{n} \left\{\sum\limits_{i=1}^n w_{in}m(O_i;\theta_0,h_0)
-M_{h_0}U_{h_0}^{-1}\sum\limits_{i=1}^n w_{in}u(O_i;h_0)\right\}+o_{P_{OW}}(1).
\end{align*}
By writing $\sqrt{n}\left(\hat{\theta}_{n,BB}-\hat{\theta}_n\right)=\sqrt{n}\left(\hat{\theta}_{n,BB}-{\theta}_0\right)-\sqrt{n}\left(\hat{\theta}_n-\theta_0\right)$, we obtain
\[
\sqrt{n}\left(\hat{\theta}_{n,BB}-\hat{\theta}_n\right)=-M_{\theta_0}^{-1}\sqrt{n}\sum\limits_{i=1}^n \left(w_{in}-\frac{1}{n}\right)\left\{m\left(O_i;{\theta}_0, h_0\right)-M_{h_0}U_{h_0}^{-1}u\left(O_i;h_0\right)\right\}+o_{P_{OW}}(1)
\]
\end{proof}

\bibliographystyle{ba}
\bibliography{references}

\end{document}